\documentclass[11pt]{amsart}
\usepackage{amsmath,amssymb,amsthm,graphicx,epstopdf,mathrsfs,url}

\usepackage{lmodern}

\def\Re{{\rm Re}\,}
\def\Im{{\rm Im}\,}

\def\arr{\rightarrow}

\newcommand{\ee}{\mathsf{e}}
\newcommand{\zz}{\mathsf{z}}

\def\tt{\theta}

\def\lm{\lambda}

\def\s{\sigma}
\def\sd{\sigma_{\rm d}}
\def\sess{\sigma_{\rm ess}}

\def\ii{{\mathsf{i}}}
\def\p{\partial} 

\def\pp{\prime}
\def\dpp{{\prime\prime}}

\def\sfH{\mathsf{H}}

\def\dd{{\mathsf{d}}}
\def\uhr{\upharpoonright}

\def\frf{{\mathfrak f}}
\def\frj{{\mathfrak j}}

\newcounter{counter_a}
\newenvironment{myenum}{\begin{list}{{\rm(\roman{counter_a})}}%
{\usecounter{counter_a}
\setlength{\itemsep}{1.ex}\setlength{\topsep}{0.8ex}
\setlength{\leftmargin}{5ex}\setlength{\labelwidth}{5ex}}}{\end{list}}

\usepackage[latin1]{inputenc}
\usepackage[T1]{fontenc}

\newcommand{\eg}{{\it e.g.}\,}
\newcommand{\ie}{{\it i.e.}\,}
\newcommand{\cf}{{\it cf.}\,}

\numberwithin{figure}{section}
\numberwithin{equation}{section}
\theoremstyle{plain}
\newtheorem*{thm*}{Theorem}
\newtheorem{thm}{Theorem}[section]
\newtheorem{hyp}{Hypothesis}[section]
\newtheorem{lem}[thm]{Lemma}
\newtheorem{prop}[thm]{Proposition}
\newtheorem{example}[thm]{Example}

\newtheorem{cor}[thm]{Corollary}

\newtheorem*{thmA}{Theorem A}
\newtheorem*{thmB}{Theorem B}

\newtheorem{dfn}[thm]{Definition}
\theoremstyle{remark}
\newtheorem{remark}[thm]{Remark}
\theoremstyle{plain}

\newcommand{\supp}{\mathrm{supp}\,}
\newcommand{\beu}{\begin{equation*}}
\newcommand{\eeu}{\end{equation*}}
\newcommand{\besu}{\begin{equation*}
\begin{aligned}}
\newcommand{\eesu}{\end{aligned}
\end{equation*}}
\newcommand{\bes}{\begin{equation}
\begin{aligned}}
\newcommand{\ees}{\end{aligned}
\end{equation}}

\newcommand\cD{\mathcal D}

\newcommand\fra{\mathfrak a}

\newcommand\eps{\varepsilon}

\newcommand\ov{\overline}
\newcommand\wt{\widetilde}
\newcommand\wh{\widehat}

\newcommand\void[1]{}

\def\ov{\overline}
\def\eps{\varepsilon}

      \def\dC{{\mathbb C}}

   \def\dN{{\mathbb N}}   
      \def\dR{{\mathbb R}}
\def\dS{{\mathbb S}}      
      
   \def\dZ{{\mathbb Z}}

\def\cD{{\mathcal D}}

\newcommand{\dom}{\mathrm{dom}\,}

\def\phi{\varphi}
\def\uhr{\upharpoonright}
\def\xx{{\mathsf{x}}}
\def\dd{{\mathsf{d}}}
\def\ii{{\mathsf{i}}}
\def\sfH{\mathsf{H}}

\def\sfT{\mathsf{T}}
\def\frf{{\mathfrak f}}
\def\frj{{\mathfrak j}}
\def\coLambda{{\Sigma \setminus \Lambda}}
\def\RcoLambda{{\dR^2\setminus\Lambda}}
\def\p{{\partial}}
\def\OpId{\mathsf{T}_{d,\omega}}
\def\OpIId{{\mathsf{H}^\Lambda_\omega}}
\def\frmId{{\fra_{d,\omega}}}
\def\frmIId{{\fra^\Lambda_\omega}}
\def\lm{{\lambda}}
\def\J{{\mathsf{J}_M}}

\def\frfR{{\frf_{D_R,\omega}^\Lambda}}

\def\bequ{\begin{equation}}
\def\eequ{\end{equation}}

\newtheorem*{openA}{Open Question A}
\newtheorem*{openB}{Open Question B}

\begin{document}
\title[{\footnotesize $\delta'$-interaction
supported on a non-closed curve in $\dR^2$}]{On absence of bound states for weakly attractive \boldmath{$\delta^\pp$}-interactions 
supported on non-closed curves in $\dR^2$}

\author{Michal Jex}
\address{Doppler Institute for Mathematical Physics and Applied
Mathematics, Czech Technical University in Prague,
B\v{r}ehov\'{a} 7, 11519 Prague, and Department of Physics,
Faculty of Nuclear Sciences and Physical Engineering, Czech
Technical University in Prague, B\v{r}ehov\'{a} 7, 11519 Prague,
Czechia}

\email{jexmicha@fjfi.cvut.cz}

\author{Vladimir Lotoreichik}
\address{Department of Theoretical Physics,
Nuclear Physics Institute, Czech Academy of Sciences, 250 68, \v{R}e\v{z} near Prague, Czechia}
\email{lotoreichik@ujf.cas.cz}

\begin{abstract}
	Let $\Lambda\subset\mathbb{R}^2$ be a non-closed
	piecewise-$C^1$ curve, which is either bounded
	with two free endpoints or unbounded with one free endpoint. 
	Let $u_\pm|_\Lambda \in L^2(\Lambda)$ be the traces
	of a function $u$ in the Sobolev space $H^1({\mathbb R}^2\setminus \Lambda)$  
	onto two faces of $\Lambda$.
	We prove that for a wide class of shapes of $\Lambda$ 
	the Schr\"odinger operator $\mathsf{H}_\omega^\Lambda$ 
	with $\delta^\prime$-interaction supported on $\Lambda$
	of strength $\omega \in L^\infty(\Lambda;\mathbb{R})$  associated with
	the quadratic form
	\[
		H^1(\mathbb{R}^2\setminus\Lambda)\ni u
		\mapsto
		\int_{\mathbb{R}^2}\big|\nabla u \big|^2 \mathsf{d} x
		 - \int_\Lambda \omega \big| u_+|_\Lambda - u_-|_\Lambda \big|^2  \mathsf{d} s
	\]
	has no negative spectrum provided that
	$\omega$ is pointwise majorized by a strictly positive function 
	explicitly expressed in terms of $\Lambda$. 
	If, additionally, the domain $\mathbb{R}^2\setminus\Lambda$ 
	is quasi-conical, we show that $\sigma(\mathsf{H}_\omega^\Lambda) = [0,+\infty)$.
	For a bounded curve $\Lambda$ in our class
	and non-varying interaction strength $\omega\in\mathbb{R}$
	we derive existence	of a constant $\omega_* > 0$ such that 
	$\sigma(\mathsf{H}_\omega^\Lambda) = [0,+\infty)$ for all $\omega \in (-\infty, \omega_*]$;
	informally speaking, bound states are absent
	in the weak coupling regime.
\end{abstract}


\maketitle

\section{Introduction}

In this paper we study the self-adjoint operator
corresponding to the formal differential expression
\[
	-\Delta - \omega\delta^\pp(\cdot - \Lambda),
	\qquad 
	\text{on}~\dR^2,
\]
with the $\delta^\pp$-interaction supported on a non-closed piecewise-$C^1$ curve $\Lambda\subset\dR^2$, which is either bounded with two free endpoints
or unbounded with one free endpoint,
here $\omega\in L^\infty(\Lambda;\dR)$ is called the strength of the interaction.
More precisely, for any function $u$ in the Sobolev space $H^1(\RcoLambda)$
its traces $u_\pm|_\Lambda$ onto two faces of $\Lambda$ 
turn out to be well-defined as functions in $L^2(\Lambda)$,
and employing the shorthand notation
$[u]_\Lambda := u_+|_\Lambda  - u_-|_\Lambda$ 
we introduce the following symmetric sesquilinear form
\begin{equation}\label{intro:form}
\begin{split}
	\fra_\omega^\Lambda[u,v] & := 
		(\nabla u, \nabla v)_{L^2(\dR^2;\dC^2)} - 
		(\omega [u]_\Lambda, [v]_\Lambda)_{L^2(\Lambda)},\\
	\dom \fra_\omega^\Lambda & := H^1(\RcoLambda),
\end{split}		 
\end{equation}
which is closed, densely defined, and semibounded
in the Hilbert space $L^2(\dR^2)$; see Proposition~\ref{prop:form}. 
Let $\OpIId$ be defined as the unique
self-adjoint operator representing the form $\fra_\omega^\Lambda$
in the usual manner. We regard $\OpIId$ 
as the \emph{Schr\"odinger 
operator with $\delta^\pp$-interaction of strength $\omega$
supported on $\Lambda$}.

The aim of this paper is to demonstrate 
a peculiar spectral property of $\OpIId$. Namely, we show 
absence of negative spectrum for $\OpIId$
under not too restrictive assumptions on the shape of $\Lambda$
and assuming that the strength $\omega$
is pointwise majorized 
by a strictly positive function explicitly
expressed in terms of the shape of $\Lambda$.
The important point to note here is that the discovered phenomenon is 
non-emergent for $\delta^\pp$-interactions supported on loops in $\dR^2$; 
\cf~\cite[Thm. 4.4]{BEL14_RMP}.

The basic geometric ingredient in our paper
is the concept of \emph{monotone curves}. 
A non-closed piecewise-$C^1$ curve $\Lambda\subset\dR^2$ is monotone if
it can be parametrized via the piecewise-$C^1$ mapping 
$\phi \colon (0, R) \arr \dR$, $R \in (0,+\infty]$, as 
\begin{equation}\label{intro:param}
	\Lambda = \big\{
				\xx_0 + (r\cos\phi(r), r\sin\phi(r)) \in \dR^2 
				\colon r \in (0,R) 
			  \big\};
\end{equation} 
here $\xx_0 \in\dR^2$ is fixed.
For example, a circular arc subtending an angle $\tt\le\pi$ is monotone,
whereas a circular arc subtending an angle $\tt >\pi$ is not.

In the next theorem, which is the first main result of our paper,
we provide a condition on $\omega$ ensuring 
absence of negative spectrum for the operator $\OpIId$ 
with $\Lambda$ being monotone.
The statement of Theorem A below is contained 
in Theorem~\ref{thm:absence2}, in Subsection~\ref{ssec:abs_main}.
\begin{thmA}
	Let a monotone piecewise-$C^1$ curve 
	$\Lambda\subset\dR^2$ be parametrized
	as in~\eqref{intro:param} via 
	$\phi \colon (0, R) \arr \dR$, $R \in (0,+\infty]$.
	Then the spectrum of $\OpIId$ satisfies
	\[
		\s(\OpIId) \subseteq [0,+\infty)\qquad \text{if}\quad
		\omega(r) \le \frac{1}{2\pi r\sqrt{1+(r\varphi^\pp(r))^2}},
		\quad\text{for}\quad r\in (0,R).
	\]		
	If $\omega$ is majorized as above and, additionally,
	the domain $\RcoLambda$ is quasi-conical, then 
	$\s(\OpIId) = [0, +\infty)$.
\end{thmA}
%
Roughly speaking, a domain $\Omega\subset\dR^2$ is quasi-conical
if it contains a disc of arbitrary large radius; see Subsection~\ref{ssec:ess}
for details. In Proposition~\ref{prop:interval} we demonstrate that, 
in general, the operator $\OpIId$ may have negative spectrum if the 
$\delta^\pp$-interaction is ``sufficiently strong''.

Operators $\OpIId$ with non-varying strengths $\omega\in\dR$
are of special interest. One can derive from Theorem A 
that for a bounded monotone $\Lambda$ one can find a constant $\omega_* > 0$
such that 
\begin{equation}\label{intro:weak}
	\s(\OpIId) = [0,+\infty)
	\qquad	\text{for all}~\omega \in (-\infty, \omega_*];
\end{equation}
in other words, there are no bound states in the weak coupling regime.
Computation of  the largest constant $\omega_* > 0$ 
such that~\eqref{intro:weak} still holds 
presents a more delicate problem, which will be considered elsewhere.

In the formulation of the second main result of the paper we use
the notion of a \emph{linear fractional transformation (LFT)}.
The complex plane $\dC$ can be extended up to the \emph{Riemann sphere}
$\wh \dC := \dC\cup \{\infty\}$ with a suitable topology
and for $a,b,c,d\in\dC$ such that $ad - bc \ne 0$ 
one defines the LFT as
\[
	M \colon \wh\dC\arr\wh\dC,\qquad M(z) := \frac{az+b}{cz+d},
\]
with the exception of the points $z = \infty$ and $z = -d/c$ if $c \ne 0$,
which have to be treated separately;
see Subsection~\ref{ssec:lft}. The next theorem generalizes 
Theorem A to the case of  curves, which are images of monotone
curves under LFTs; the statement of this theorem is contained 
in Theorem~\ref{thm:absenceLFT}, in Subsection~\ref{ssec:abs_LFT}.
Here, we confine ourselves to non-varying interaction strengths only.
\begin{thmB}
	Let $\Lambda \subset \dR^2$ be a bounded piecewise-$C^1$ curve. 
	Suppose that there exists an LFT 
	$M\colon \wh\dC\arr\wh\dC$ such that
	$M(\infty), M^{-1}(\infty) \notin \Lambda$ and that
	$M^{-1}( \Lambda )$ is monotone.
	Then there exists a constant $\omega_* > 0$ such that
	\[
		\s(\OpIId) = [0,+\infty)
		\qquad
		\text{for all}~\omega \in (-\infty, \omega_*].
	\]	
\end{thmB}
In the main body of the paper also an explicit formula for $\omega_*$ 
in the above theorem is provided. Using Theorem~B we can treat, \eg, 
any circular arc, since
it can be mapped via a suitable LFT to a subinterval of the 
straight line in $\dR^2$; see Example~\ref{example3}. 
One may even conjecture that for any bounded $\Lambda$ 
there exists an $\omega_* > 0$ such that
$\s(\OpIId) = [0,+\infty)$ for all $\omega  \in (-\infty, \omega_*]$.

Our proofs rely on the min-max principle 
applied to the form $\fra_\omega^\Lambda$ in~\eqref{intro:form}
on a suitable core. A further important ingredient in our
analysis is careful investigation of a  model
one-dimensional problem with a point $\delta^\pp$-interaction
on the loop. 

The results of this paper contribute to a prominent
topic in spectral theory:
existence/non-existence of weakly coupled bound
states for Schr\"odinger-type operators.
Absence of bound states
in the weak coupling regime holds for 
Schr\"odinger operators with regular potentials in space dimensions
$d \ge 3$  (but not for $d = 1,2$!); see \cite{S76}. Also such an effect
occurs for $\delta$-interactions supported on arbitrary compact hypersurfaces
in $\dR^3$ (see \cite{EF09}) and for $\delta$-interactions on 
compact non-closed curves in $\dR^3$ (see \cite{EK08}). However, for
$\delta$-interactions in $\dR^2$ supported on arbitrary compact curves
such an effect is non-existent~\cite{ET04,KL14}.

Schr\"odinger operators with $\delta^\pp$-interactions
supported on hypersurfaces are attractive from physical
point of view, because they exhibit rather unusual scattering
properties; \cf \cite[Chap. I.4]{AGHH05}.
These operators are also physically relevant
in photonic crystals theory~\cite{FK96}.
As a mathematical abstraction they were perhaps first studied 
in~\cite{AGS87, S88}, where interactions were supported on spheres.
A rigorous definition of such operators
with interactions supported on general hypersurfaces has been posed 
in~\cite[\S 7.2]{E08} as an open question.
Such Hamiltonians with interactions supported on closed hypersurfaces 
without free boundaries have been rigorously defined in~\cite{BLL13_AHP} 
using two approaches: via the theory of self-adjoint extensions
of symmetric operators and by means of form methods. 
Spectral properties of them were investigated
in several subsequent 
works~\cite{BEL14_RMP, BGLL15, EJ13, EJ14, EKh15, J15, LR15}. 
In the recent preprint~\cite{MPS15}
Schr\"odinger operators with $\delta^\pp$-interactions supported on non-closed curves and surfaces 
are defined via the theory of self-adjoint extensions and
their scattering properties are discussed.

Let us briefly outline the structure of the paper. 
Section~\ref{sec:prelim} presents some preliminaries:
Sobolev spaces, geometry of curves, linear fractional
transformations, and a model one-dimensional spectral problem. 
Section~\ref{sec:def} provides a rigorous definition of the
operator $\OpIId$ and a characterisation of its essential spectrum.
Section~\ref{sec:abs} contains proofs of our main results,
formulated in Theorems A and B, as well as some related results 
and examples. In Section~\ref{sec:open} final remarks are given
and two open questions are posed.
A couple of standard proofs of identities related
to LFTs are outsourced to Appendix~\ref{app:A}.

\subsection*{Notations}
By $D_R(\xx) := \{\xx\in \dR^2\colon |\xx - \xx_0| < R\}$ 
we denote the open disc of
the radius $R > 0$ with the center $\xx_0 \in\dR^2$. 
If such a disc is centered at the origin,
we use the shorthand notation $D_R := D_R(0)$. 
By definition we set $D_\infty := \dR^2$. 
For a self-adjoint operator $\sfT$ we denote by 
$\sess(\sfT)$, $\sd(\sfT)$, and $\s(\sfT)$  its essential, discrete,
and full  spectra, respectively. For an open set $\Omega\subset\dR^2$
the space of smooth compactly supported functions
and the first order order Sobolev space are denoted by
$\cD(\Omega)$ and by $H^1(\Omega)$, respectively.

\subsection*{Acknowledgements}

The authors are indebted to Jussi Behrndt and Pavel Exner
for their active interest in the preparation of this paper
and for many stimulating conversations. David Krej\v{c}i\v{r}\'{i}k
is acknowledged for a comment, which led to an improvement in the
definition of the operator.
MJ was supported by 
the Grant Agency of the Czech Technical University in Prague, 
grant No. SGS13/217/OHK4/3T/14.
VL was supported by the Austrian Science Fund (FWF)
under the project P~25162-N26.
Both the authors acknowledge the financial support by 
the Austria-Czech Republic co-operation grant  CZ01/2013
and by the Czech Science Foundation (GA\v{C}R)
under the project 14-06818S.

\section{Preliminaries}\label{sec:prelim}

This section contains some preliminary material that will be
used in the main part of this paper. In Subsection~\ref{ssec:Sobolev}
we provide basic facts on the Sobolev space $H^1$, in particular,
we define the Sobolev space $H^1(\RcoLambda)$ for a
non-closed Lipschitz curve $\Lambda$. In Subsection~\ref{ssec:curves} we
introduce the notions of a piecewise-$C^1$ curve and of a monotone curve. 
The concept of the linear fractional
transformation is discussed in Subsection~\ref{ssec:lft}.
A model spectral problem for one-dimensional Schr\"odinger operator
with one-center $\delta^\pp$-interaction on a loop is considered
in Subsection~\ref{ssec:modelproblem} and 
a sufficient condition for absence of negative eigenvalues 
in this spectral problem is established.

\subsection{Sobolev spaces}\label{ssec:Sobolev}

Let $\Omega \subset\dR^2$ be a simply connected 
\emph{Lipschitz domain} from the class 
described in~\cite[Ch. VI]{St}. This class of Lipschitz domains
includes (as a subclass) Lipschitz domains with compact boundaries
as in~\cite[Ch. 3]{McL}, \emph{hypographs} of uniformly
Lipschitz functions, and some other domains with non-compact
boundaries. In what follows the Hilbert spaces $L^2(\Omega)$,
$L^2(\Omega;\dC^2)$, $L^2(\partial\Omega)$, and $H^1(\Omega)$ are
defined as usual; see \eg \cite[Ch. 3]{McL} and \cite{M87}.
For the sake of brevity we denote the scalar products in both
$L^2(\Omega)$ and $L^2(\Omega;\dC^2)$ by $(\cdot,\cdot)_\Omega$
without any danger of confusion. The scalar product in $L^2(\p\Omega)$
is abbreviated by $(\cdot,\cdot)_{\p\Omega}$.
The space of functions on $\ov\Omega$ smooth up to the boundary
$\p\Omega$ is defined as
\[
	\cD(\ov\Omega) := \big\{ u|_{\Omega} \colon u\in \cD(\dR^2)\big\}.
\]
By~\cite[Thm. 3.29]{McL}, see also \cite{M87, St}, the space
$\cD(\ov\Omega)$ is dense in both $L^2(\Omega)$ and $H^1(\Omega)$.
The natural restriction mapping 
$\cD(\ov\Omega)\ni u\mapsto u|_{\p\Omega}\in L^2(\p\Omega)$ 
can be extended by continuity up to the whole
space $H^1(\Omega)$; see \eg \cite[Thm. 3.37]{McL} and \cite{M87}.
The corresponding extension by continuity 
$H^1(\Omega)\ni u \mapsto u|_{\p\Omega} \in L^2(\p\Omega)$
is called \emph{the trace mapping}.
The statement of the first lemma in this subsection
appears in several monographs and papers in various forms; 
see \eg \cite[Lem. 2.6]{BEL14_RMP} and \cite[Lem. 2.5]{GM08}
for two different proofs of this statement.

\begin{lem}\label{lem:BEL} 
	Let $\Omega\subset\dR^2$ be a Lipschitz domain. Then
	for any $\eps > 0$ there exists a constant $C(\eps) > 0$ 
	such that
	\[
		\big\|u|_{\p\Omega}\big\|^2_{\p\Omega} 
		\le 
		\eps\|\nabla u\|^2_\Omega + C(\eps)\|u\|^2_\Omega
	\]
	holds for all $u\in H^1(\Omega)$.
\end{lem}

The following hypothesis will be used throughout the  paper.
\begin{hyp}\label{hyp:curve}
	Let $\Omega_+\subset\dR^2$ 
	be a simply connected Lipschitz domain from the above class,
	whose complement $\Omega_- := \dR^2\setminus\ov\Omega_+$ 
	is a Lipschitz domain from the same class.
	Set $\Sigma := \p\Omega_+ = \p\Omega_-$
	and suppose that $\Lambda\subset\Sigma$ is a connected subarc 
	of $\Sigma$, which is not necessarily bounded if $\Sigma$ is unbounded.
\end{hyp} 

Obviously, the orthogonal sum $H^1(\Omega_+)\oplus H^1(\Omega_-)$ 
is a Hilbert space with respect to the scalar product
\[
	(u_+\oplus u_-,v_+\oplus v_- )_1 
	:=
	(u_+,v_+)_{H^1(\Omega_+)} + (u_-,v_-)_{H^1(\Omega_-)}, 
	\quad  
	u_\pm, v_\pm
	\in
	H^1(\Omega_\pm).
\]
The norm associated to this scalar product is denoted by
$\|\cdot\|_1$. Let us define the jump of the trace as
\[
	[u]_\Sigma := u_+|_\Sigma - u_-|_\Sigma,
	\qquad 
	u = u_+\oplus u_-\in H^1(\Omega_+) \oplus H^1(\Omega_-).
\]
The Hilbert space $L^2(\Sigma)$ can be decomposed into the orthogonal sum 
\[
	L^2(\Sigma) = L^2(\Lambda)\oplus L^2(\coLambda).
\]
The scalar products in $L^2(\Lambda)$ and $L^2(\coLambda)$
will further be
denoted by $(\cdot,\cdot)_\Lambda$ and $(\cdot,\cdot)_\coLambda$. 
Clearly enough, the restrictions of $u_\pm|_\Sigma$ for a
$u_\pm \in H^1(\Omega_\pm)$ to the arcs $\coLambda$ and $\Lambda$
satisfy $u_\pm|_\coLambda \in L^2(\coLambda)$ and 
$u_\pm|_{\Lambda} \in L^2(\Lambda)$. 
Let us also introduce the notations
\[
	[u]_\bullet := u_+|_\bullet - u_-|_\bullet,
	\quad
	\bullet = \Lambda, \coLambda,
	\quad 
	u = u_+ \oplus u_- \in H^1(\Omega_+)\oplus H^1(\Omega_-).
\]	
The linear space
\begin{equation}\label{FLambda} 
	F_\Lambda := 
	\big\{
		u\in\cD(\ov{\Omega_+})\oplus\cD(\ov{\Omega_-})
		\colon [u]_\coLambda = 0
	\big\}
\end{equation}
is a subspace of the Hilbert space 
$H^1(\Omega_+)\oplus H^1(\Omega_-)$,
and its closure in  $H^1(\Omega_+)\oplus H^1(\Omega_-)$
\begin{equation}\label{H1} 
	H^1(\RcoLambda) := \ov{F_\Lambda}^{\|\cdot\|_1}	
\end{equation}
is itself a Hilbert space with respect to the same scalar product 
$(\cdot,\cdot)_1$.

\begin{remark}
	The above construction of the space $H^1(\RcoLambda)$ can
	easily be translated to the higher space dimensions, 
	in which case $\Lambda$ will be a hypersurface 
	with free boundary (open hypersurface).
\end{remark}

\begin{remark}
	The space $H^1(\RcoLambda)$ can also be defined in an
	alternative way. The set $\RcoLambda$ is an open subset
	of $\dR^2$. Hence, one can define for any $u\in L^2(\dR^2)$ 
	its weak partial derivatives $\p_1 u$ and $\p_2 u$ 
	by means of the test functions in $\cD(\RcoLambda)$; 
	see \eg \cite[Ch. 3]{McL}. Then the space $H^1(\RcoLambda)$ is given by
	\[
		H^1(\RcoLambda) = 
		\big\{u\in L^2(\dR^2)\colon \p_1 u, \p_2 u\in L^2(\dR^2)\big\}.
	\]
	We are not aiming to provide here an argumentation that 
	this new definition gives rise to the same space as in~\eqref{H1}. 
	It is only important here that the equivalence
	of these definitions automatically implies that the space
	$H^1(\RcoLambda)$
	is independent of the continuation of the arc $\Lambda$ up to $\Sigma$.
	Another way of verifying the independence of the space 
	$H^1(\RcoLambda)$
	from a continuation of $\Lambda$ can be found in~\cite{D12}.
\end{remark}

Next proposition collects some useful properties of the traces of
functions in $H^1(\RcoLambda)$ onto $\coLambda$ and onto $\Lambda$.

\begin{prop}\label{prop:trace}
	Let the curves $\Sigma, \Lambda\subset\dR^2$, and the domains 
	$\Omega_\pm\subset\dR^2$ be as in Hypothesis~\ref{hyp:curve}. 
	Let the Hilbert space $(H^1(\RcoLambda), (\cdot,\cdot)_1)$ be 
	as in~\eqref{H1}. Then the following statements hold.
	\begin{myenum}
		\item $[u]_\coLambda = 0$ for all $u\in H^1(\RcoLambda)$.
		\item For any $\eps >0$ there exists a constant $C(\eps) >0$
			  such that
		  	\[
		  		\big\| [u]_\Lambda \big \|^2_\Lambda
		  		\le 
			  	\eps\|\nabla u\|^2_{\dR^2} + C(\eps)\|u\|^2_{\dR^2}
			  \]
		  for all $u\in H^1(\RcoLambda)$.
	\end{myenum}
\end{prop}

\begin{proof}
	(i) 
	It can be easily checked that the continuity of the trace mappings
	\[
		H^1(\Omega_\pm)\ni u_\pm \mapsto u_\pm|_\Sigma \in L^2(\Sigma)
	\]
	implies that the mapping
	\[
		H^1(\Omega_+)\oplus H^1(\Omega_-)\ni u \mapsto [u]_\coLambda
		\in L^2(\coLambda)
	\]
	is well-defined and continuous.
	Note that for any $u\in H^1(\RcoLambda)$ there exists an
	approximating sequence $(u_n)_n\subset F_\Lambda$ (\cf \eqref{H1}) 
	such that $\|u_n - u\|_1\arr 0$ as 
	$n\arr \infty$. Hence, we obtain
	\[	
		[u]_\coLambda = \lim_{n\arr\infty}[u_n]_\coLambda  = 0.
	\]		

	(ii) By Lemma~\ref{lem:BEL} for any $\eps > 0$ there exist
	constants $C_\pm(\eps) > 0$ such that
	\begin{equation}\label{bounds} 
		\|u_\pm|_\Sigma\|^2_\Sigma
		\le 
		(\eps/2)\|\nabla u_\pm\|^2_{\Omega_\pm}
		+
		C_\pm(\eps)\|u_\pm\|^2_{\Omega_\pm}
	\end{equation}
	for all $u \in H^1(\Omega_+)\oplus H^1(\Omega_-)$. 
	Set then $C(\eps) := \max\{2C_+(\eps),2C_-(\eps)\}$.
	Using the result of item (i) and the bounds~\eqref{bounds} we obtain
	that for any $\eps > 0$ and any 
	$u = u_+\oplus u_-\in H^1(\RcoLambda)
	\subset H^1(\Omega_+)\oplus H^1(\Omega_-)$
	holds
	\[
	\begin{split}
		\big\|[u]_\Lambda\big\|^2_\Lambda 
		& = \big\|[u]_\Sigma\big\|^2_\Sigma \le 	
		     2\|u_+|_\Sigma\|^2_\Sigma + 
			 2\|u_-|_\Sigma\|^2_\Sigma\\
		& \le \eps\|\nabla u_+\|_{\Omega_+}^2
				+\eps\|\nabla u_-\|_{\Omega_-}^2
				+ 2C_+(\eps)\|u_+\|^2_{\Omega_+} 
				+ 2C_-(\eps)\|u_-\|^2_{\Omega_-}\\
		& \le \eps\|\nabla u\|^2_{\dR^2}
				+ C(\eps)\|u\|^2_{\dR^2}.
	\end{split}
	\]
\end{proof}

\begin{remark}
	For $\omega_1,\omega_2\in L^\infty(\Lambda;\dR)$
	by writing $\omega_1 \le \omega_2$ we will always implicitly mean
	that $\omega_2 - \omega_1 \ge 0$ almost everywhere.
\end{remark}

\subsection{On curves in $\dR^2$}\label{ssec:curves}

We begin this subsection by defining the notion of a piecewise-$C^1$ curve.
It should be emphasized that, especially for unbounded curves,
definition of a piecewise-$C^1$ curve is non-unique
in the mathematical literature.

\begin{dfn}\label{def:C1}
	A non-closed curve $\Lambda\subset \dR^2$
	satisfying Hypothesis~\ref{hyp:curve} is called  \emph{piecewise-$C^1$}
	if it can be parametrized via a piecewise-$C^1$ mapping 
	\begin{equation}\label{lambda}
		\lm\colon I\arr \dR^2, \qquad \lm(s) := (\lm_1(s),\lm_2(s)),
		\quad I := (0,L),~L \in (0,+\infty],
	\end{equation}
	such that $\lm(I) = \Lambda$ and $\lm$ is injective. 
	If, moreover, $|\lm^\pp(s)|= 1$ for almost all
	$s\in I$, then such a parametrization is called~\emph{natural}
	and $L$ is then called the length of $\Lambda$.
\end{dfn}
We require in the above definition, that the curve $\Lambda$
satisfies Hypothesis~\ref{hyp:curve}, 
to avoid increasing oscillations at infinity for unbounded curves.

Further, we  proceed to define a (non-standard)
concept of a~\emph{monotone curve}.
The authors have not succeeded to find a common name for this concept
in the literature on geometry.

\begin{dfn}\label{def:monotone}
	A piecewise-$C^1$ curve $\Lambda\subset \dR^2$
	is called \emph{monotone} if it can be parametrized 
	via a piecewise-$C^1$ mapping 
	$\phi\colon (0,R)\arr \dR$ with $R \in (0,+\infty]$ such that
	\[
		\Lambda 
		= 
		\big\{
			\xx_0 + (r\cos\phi(r), r\sin\phi(r)) \in\dR^2
			\colon r\in (0,R)
		\big\}
	\]
	with some fixed $\xx_0\in\dR^2$.	
\end{dfn}

Informally speaking, 
a curve $\Lambda$ is monotone if the distance (measured
in $\dR^2$) from one of its endpoints is always increasing
when travelling along $\Lambda$ from this endpoint
towards another endpoint or towards infinity.

\begin{remark}\label{rem:omegar}
	For a curve $\Lambda\subset\dR^2$ as in Definition~\ref{def:monotone}
	any function $\omega\in L^\infty(\Lambda)$ can be viewed
	as a function of the argument $r\in (0,R)$.
\end{remark}

\subsection{Linear fractional transformations}\label{ssec:lft}

For later purposes we introduce linear fractional transformations (LFT)
and state several useful properties of them. 
To work with LFT it is more convenient to deal with the
\emph{extended complex plane (Riemann sphere)} $\wh\dC := \dC\cup\{\infty\}$ 
rather than the usual complex plane. The complex plane itself
as a subset of $\wh\dC$ can be naturally identified with the Euclidean
plane $\dR^2$ and occasionally we will use this identification. 

For the purpose of convenience the extended complex plane $\wh\dC$ 
is endowed with a suitable topology: 
a sequence $(z_n)_n \in \wh\dC$ converges to $z\in\wh\dC$ if 
one of the following  conditions holds:
\begin{myenum}
	\item	$z = \infty$ and there exists $N\in\dN$ such that
			$z_n = \infty$ for all $n \ge N$;
	\item 	$z = \infty$ and any infinite
			subsequence $(z_{n_k})_k\subset \dC$ of $(z_n)_n$ satisfies	
			$\lim\limits_{k \arr \infty} |z_{n_k}| = \infty$;
	\item	$z \in\dC$, there exists $N\in\dN$ such that
			$z_n \ne \infty$ for all $n \ge N$, 
			and $\lim\limits_{n\arr \infty} z_n = z$ 
			in the sense of convergence in $\dC$.
\end{myenum}			
This definition of topology can also be easily reformulated in terms of
open sets.
The above topology on $\wh\dC$ is equivalent to the topology of $\dS^2$ 
(unit sphere in $\dR^3$). 
A natural homeomorphism between $\wh\dC$ and $\dS^2$ is called		
\emph{stereographic projection}; see \eg \cite[\S 6.2.3]{Krantz}.
%

For $a,b,c,d\in\dC$ such that $ad - bc \ne 0$ the mapping
$M\colon \wh\dC\arr\wh\dC$ is an LFT if one of the two 
conditions holds:
\begin{myenum}
	\item	$c = 0$, $d \ne 0$, $M(\infty) := \infty$,
			and $M(z) := (a/d)z + (b/d)$ for $z\in\dC$.   
	\item	$c \ne 0$, $M(\infty) := a/c$, $M(-d/c) := \infty$,
			and $M(z) := \frac{az+b}{cz+d}$ for $z\in\dC$, $z\ne -d/c$.   
\end{myenum}			
The following  statement can be found in~\cite[\S 6.2.3]{Krantz}.
%
\begin{prop}\label{prop:LFT}
	Any LFT $M\colon \wh\dC\arr\wh\dC$  
	is a homeomorphism with respect to the above topology on $\wh\dC$
	and its inverse $M^{-1}$ is also an LFT.
	The composition $M_1\circ M_2$ of two LFTs $M_1,M_2$ is an LFT as well.
\end{prop}
%
It is convenient to introduce
$M_1(x,y) := \Re M(x + \ii y)$ and $M_2(x,y) := \Im M(x+ \ii y)$.
Then \emph{Cauchy-Riemann equations}
\begin{equation}\label{CR}
	\p_x M_1 = \p_y M_2,\qquad \p_x M_2 = -\p_y M_1, 
\end{equation}
hold pointwise in $\dR^2$ except the point $M^{-1}(\infty)$.
In view of these equations the Jacobian $\J$ of
the mapping $M$ can be computed (again except the point $M^{-1}(\infty)$) 
by the formulae
\begin{equation}\label{J}
	\J = (\p_x M_1)^2 + (\p_y M_1)^2 = (\p_x M_2)^2 + (\p_y M_2)^2;
\end{equation}
also the following relation turns out to be useful
\begin{equation}\label{relation}
	\langle\nabla M_1,\nabla M_2\rangle
			= \p_x M_1\p_x M_2 + \p_y M_1\p_y M_2 = 0;
\end{equation}
\ie the vectors $\nabla M_1$ and $\nabla M_2$
are orthogonal to each other.

Next auxiliary lemma is of purely technical nature
and is proven for convenience of the reader.

\begin{lem}\label{lem:M}
	Let $M$ be an LFT with the Jacobian $\J$. 
	Then for any $x\in\dR^2$, $x \ne M^{-1}(\infty)$, 
	and any function $u \colon\dR^2\simeq\dC\arr \dC$ 
	differentiable at the point $M(x)$
	\[
		|(\nabla v)(x)|^2 = |(\nabla u)(M(x))|^2 \J(x)
	\]
	holds with $v = u\circ M$. 	
\end{lem}

\begin{proof}
	Using relations~\eqref{J},~\eqref{relation},
	and the chain rule for differentiation we obtain
	\[
	\begin{split}
		|\nabla v|^2 & = 
		\big|(u_x^\pp\circ M)\p_x M_1+(u_y^\pp\circ M)\p_x M_2\big|^2
		+
		\big|(u_x^\pp\circ M)\p_y M_1 + (u_y^\pp\circ M)\p_y M_2\big|^2
		\\[0.5ex]
		& = 
		\Big(|u_x^\pp\circ M|^2 + |u_y^\pp\circ M|^2\Big)\J 
		+	2\Re \Big[ \big((u^\pp_x u^\pp_y)\circ M\big)\cdot
		\big\langle \nabla M_1, \nabla M_2 \big\rangle \Big]\\[0.5ex]
		&= 
		\Big(|u_x^\pp\circ M|^2 + |u_y^\pp\circ M|^2\Big)\J
		= 
		|(\nabla u)\circ M|^2 \J.
	\end{split}
	\]
	The claim is thus shown.
\end{proof}

\subsection{Point $\delta'$-interaction on a loop}
\label{ssec:modelproblem}

In this subsection we introduce an auxiliary self-adjoint 
Schr\"o\-dinger operator $\OpId$ acting in the Hilbert space 
$( L^2(I), (\cdot,\cdot)_I )$ with $I := (0,d)$ and 
corresponding to a point $\delta^\pp$-interaction
on the one-dimensional loop of length $d > 0$. 
Employing the following shorthand notation
\[
	[\psi]_{\p I} := \psi(d-) - \psi(0+),\qquad\psi \in H^2(I),
\]
we define
\begin{equation}\label{def:opId}
	\OpId\psi := -\psi^\dpp,
	\quad 
	\dom \OpId := \big\{\psi\in H^2(I)\colon 
		\psi^\pp(0+) = \psi^\pp(d-) = \omega [\psi]_{\p I}\big\},
\end{equation}
where $\omega\in\dR$; see \cite{AGHH05, BN13, EKMT14, ENZ01, GM09, KM10} 
for the investigations of more general
operators of this type. Note that $\omega = 0$  
corresponds to Neumann boundary conditions at the endpoints.
Next proposition states a spectral property of $\OpId$, which is
useful for our purposes.

\begin{prop}\label{prop:loop} 
	The self-adjoint operator $\OpId$ in the Hilbert space $L^2(I)$, 
	defined in~\eqref{def:opId}, is non-negative if $d\omega \le 1$.
\end{prop}
\begin{proof}
	We prove this proposition via construction of an
	explicit condition for
	the negative spectrum of $\OpId$ and its analysis.
	Obviously, the spectrum of $\OpId$ is discrete 
	(due to the compact embedding of $H^2(I)$ into $L^2(I)$). 
	An eigenfunction of $\OpId$, which corresponds to a negative
	eigenvalue $\lm = -\kappa^2 < 0$ ($\kappa > 0$) is characterized by the
	following two conditions:
	\begin{subequations}
	\begin{align}
		\label{c1} -\psi^\dpp(x)& =-\kappa^2\psi(x);\\
		\label{c2} \psi^\pp(0 +) & = \psi^\pp(d-) = \omega[\psi]_{\p I}.
	\end{align}
	\end{subequations}
	The condition~\eqref{c1} is satisfied by a function, 
	which can be represented in the form
	\begin{equation*}
		\psi(x) = A \exp(\kappa x)+B \exp(-\kappa x),\qquad x\in(0,d),
	\end{equation*}
	with some constants $A, B\in\dC$. Simple computations yield
	\begin{align*}
		&\psi(0+) = A + B,
		&\psi(d-) = A\exp(\kappa d)  + B\exp(-\kappa d),\\
		&\psi^\pp(0+) = \kappa A - \kappa B,
		&\psi^\pp(d-) = \kappa A\exp(\kappa d)  - \kappa B\exp(-\kappa d).
	\end{align*}
	The above identities and the condition~\eqref{c2} together imply
	\begin{subequations}	
	\begin{align}
		\label{AB} 
		& A = \frac{1 -\exp(-\kappa d)}{1- \exp(\kappa d)}B;\\[1.2ex]
		\label{AB2} 
		&\kappa A -\kappa B = \omega\Big(A\big(\exp(\kappa d)-1\big) 
		 			+	B\big(\exp(-\kappa d) - 1\big)\Big).
	\end{align}
	\end{subequations}
	Substituting the formula~\eqref{AB} into~\eqref{AB2}, we arrive at
	\begin{equation*}
		\kappa B\Bigg(\frac{1 -\exp(-\kappa d)}{1- \exp(\kappa d)} -1\Bigg)
		= 
		\omega\Big(-B(1-\exp(-\kappa d)) + B(\exp(-\kappa d) - 1)\Big),
	\end{equation*}
	that is equivalent to
	\begin{equation*}
		\exp(-\kappa d) - \exp(\kappa d) = 
		\frac{2\omega}{\kappa}\Big(1-\exp(-\kappa d)\Big)
		\Big(1-\exp(\kappa d)\Big).
	\end{equation*}
	Making several steps further in the computations, we obtain
	\begin{equation*}
	\begin{split}
		1 & = \frac{2\omega}{\kappa}
		\frac{(1 -\exp(-\kappa d))(1-\exp(\kappa d))}
				{\exp(\kappa d)(\exp(-2\kappa d)-1)} 
			= \frac{2\omega}{\kappa}
			\frac{\exp(\kappa d)-1}{\exp(\kappa d)(\exp(-\kappa d)+1)}\\
		 & = \frac{2\omega}{\kappa} \frac{1-\exp(-\kappa d)}
		 		{1 + \exp(-\kappa d)}.
	\end{split}
	\end{equation*}
	Define then the following function
	\[
		\Theta_\omega(\kappa) := 
		\frac{2\omega}{\kappa}\frac{1-\exp(-\kappa d)}{1 + \exp(-\kappa d)}, 			
		\qquad \kappa >0.
	\]
	Hence, the point $\lm = -\kappa^2$ is a negative eigenvalue
	of $\OpId$ if and only if $\Theta_\omega(\kappa) = 1$. 
	Let us consider the following auxiliary function
	\[
		f(x) := \frac{1- \ee^{-x}}{1+\ee^{-x}},\qquad x \ge 0,
	\]
	which is clearly continuously differentiable, 
	and whose derivative is given by
	\[
		f^\pp(x) = \frac{2}{(\ee^{x/2}+\ee^{-x/2})^2}, \qquad x \ge 0.
	\]  
	Hence, using the standard inequality $a + 1/a > 2$,
	$a \in (0,+\infty)$, $a \ne 1$, we get $f^\pp(x) < 1/2$ for all $x > 0$. 
	Applying the mean value theorem to $f$, we obtain
	\[
		f(x) = f(0) + f^\pp(\xi)(x - 0) = f^\pp(\xi)x < \frac{x}{2};
	\]
	here $\xi \in (0,x)$. Finally, note that
	\[
		0\le \Theta_\omega(\kappa) = \frac{2\omega}{\kappa}f(\kappa d) 
								< d\omega .
	\]
	Thus, for $d\omega \le 1$ the equation $\Theta_\omega(\kappa) = 1$
	has no positive roots and the claim follows.
\end{proof}

According to \eg \cite{KM14},
the operator $\OpId$ represents the sesquilinear form
\begin{equation}\label{pointform}
	\frmId [\psi,\phi] 
	:= 
	(\psi^\pp,\phi^\pp)_I- 
			\omega [\psi]_{\p I} [\ov\phi]_{\p I},
			\qquad
	\dom\frmId :=  H^1(I),
\end{equation}
and we can derive the following simple corollary 
of Proposition~\ref{prop:loop}.
\begin{cor}\label{cor:nonnegative} 
	Let the sesquilinear form $\frmId$ 
	be as in~\eqref{pointform}. If $d\omega \le 1$, then 
	$\frmId [ \psi ] \ge 0$ for all $\psi\in H^1(I)$.
\end{cor}

\begin{remark}
	Consider the non-negative symmetric operator 
	\[
		S\psi := -\psi^\dpp,\qquad \dom S := H^2_0(I),
	\]
	in $L^2(I)$. The operator $S$ is known to have deficiency indices $(2,2)$. 
	One may consider self-adjoint extensions of $S$ in $L^2(I)$.  
	The self-adjoint operator $\OpId$ with $d\omega = 1$
	turns to be the \emph{Krein-von Neumann extension} of $S$ 
	(the ``smallest'' non-negative self-adjoint extension of $S$); 
	\ie for any other non-negative self-adjoint extension $T$ of $S$
	\[
		(T+a)^{-1} \le (\OpId + a)^{-1}
	\]
	holds for all $a > 0$; see \eg 
	\cite[Cor. 10.13, Thm. 14.25, Ex. 14.14]{S}.
\end{remark}

\section{Definition of the operator and its essential spectrum}
\label{sec:def}

In this section we  rigorously define using form methods Schr\"odinger operators
with $\delta^\pp$-interactions supported on non-closed curves
as in Hypothesis~\ref{hyp:curve} 
and characterise their essential spectra. In the latter characterisation
the notion of a quasi-conical domain plays an essential role.

\subsection{Definition of the operator via its sesquilinear form}

Schr\"odinger operators with $\delta^\pp$-interactions 
supported on closed hypersurfaces were defined and investigated 
in~\cite{AGS87, BEL14_RMP, BLL13_AHP, EJ13, EJ14, EKh15}.
The goal of this subsection is to define 
rigorously Schr\"odinger operator with $\delta^\pp$-interactions
supported on a non-closed curve $\Lambda$ satisfying 
Hypothesis~\ref{hyp:curve}. In the case of a bounded $C^{2,1}$-smooth
curve $\Lambda$ our definition of the operator agrees 
with the one in the recent preprint~\cite{MPS15}, 
where this Hamiltonian
is defined using the theory of self-adjoint extensions of 
symmetric operators.

Let $\omega \in L^\infty(\Lambda;\dR)$ 
and denote by $\|\omega\|_\infty$ its sup-norm.
Recall the definition of the sesquilinear form $\frmIId$
in~\eqref{intro:form}
\begin{equation}\label{eq:form}
	\frmIId[u,v] := 
	(\nabla u,\nabla v)_{\dR^2} - (\omega [u]_\Lambda, [v]_\Lambda)_\Lambda,
	\qquad
	\dom\frmIId  := H^1(\RcoLambda).
\end{equation}
If $\omega \equiv 0$, we occasionally write
$\fra_{\rm N}^\Lambda$ instead of $\frmIId$. 

\begin{prop}\label{prop:form}
	Let $\Lambda\subset\dR^2$ be as in Hypothesis~\ref{hyp:curve}, 
	let $\omega \in L^\infty(\Lambda;\dR)$, 
	and let the linear space $F_\Lambda$ be as in~\eqref{FLambda}. 
	Then the sesquilinear form $\frmIId$ in~\eqref{eq:form} 
	is closed, densely defined, symmetric, and lower-semibounded 
	in the Hilbert space $L^2(\dR^2)$.
	Moreover, $F_\Lambda\subset\dom\frmIId$ is a core for this form.
\end{prop}

\begin{proof}
	Since $\frmIId [u, u] \in\dR$ for all $u\in\dom\frmIId$, 
	the form $\frmIId$ is, clearly, symmetric. 
	It is straightforward to see the chain of inclusions 
	$\cD(\dR^2) \subset F_\Lambda\subset \dom\frmIId$. 
	Density of $\dom\frmIId$ in $L^2(\dR^2)$ 
	follows thus from the density of 
	$\cD(\dR^2)$ in $L^2(\dR^2)$; 
	for the latter see \eg \cite[Cor. 3.5]{McL}.

	The norm induced in the conventional way by the form 
	$\fra_{\rm N}^\Lambda$ on its domain $H^1(\RcoLambda)$
	is easily seen to be equivalent to the norm $\|\cdot\|_1$ introduced in
	Subsection~\ref{ssec:Sobolev}. Hence, the form $\fra_{\rm N}^\Lambda$ 
	is closed and the space $F_\Lambda$, being dense in 
	$H^1(\RcoLambda)$, is a core for it; 
	\cf \cite[Dfn. 10.2]{S}. Let us then introduce an auxiliary form 
	\[
		\fra'[u , v] := (\omega [u]_\Lambda, [v]_\Lambda)_\Lambda,
		\qquad
		\dom\fra' := H^1(\RcoLambda).
	\]
	Using Proposition~\ref{prop:trace}\,(ii) we get for all $\eps > 0$ 
	the following bound
	\[
		|\fra'[u,u]| \le 
		\eps\,\|\omega\|_\infty\fra_{\rm N}^\Lambda[u,u] 
		+ 
		C(\eps)\|\omega\|_\infty\|u\|^2_{\dR^2} 
	\]
	with some $C(\eps) > 0$. 
	Choosing $\eps < \frac{1}{\|\omega\|_\infty}$
	in the above bound, we obtain that $\fra'$ is relatively bounded
	with respect to $\fra_{\rm N}^\Lambda$ with form bound $< 1$. 
	Hence, by~\cite[Thm. 10.21]{S} \emph{(KLMN theorem)}
	the form $\frmIId = \fra_{\rm N}^\Lambda + \fra'$ 
	is closed and the space $F_\Lambda$ 
	is a core for it.
\end{proof}

\begin{dfn}\label{dfn:operator}
	The self-adjoint operator $\OpIId$ in $L^2(\dR^2)$ 
	corresponding to the form $\frmIId$
	via the first representation theorem 
	(see \eg \cite[Ch. VI, Thm. 2.1]{Kato})
	is called Schr\"odinger operator with $\delta^\pp$-interaction 
	of strength $\omega$ supported on $\Lambda$.
\end{dfn}

If $\omega$ is a non-negative function, then we occasionally say
that the respective $\delta^\pp$-interaction is \emph{attractive}. 

\subsection{Essential spectrum}\label{ssec:ess}

In this subsection we characterise 
the essential spectrum of the operator $\OpIId$. 
To this aim we require the following auxiliary lemma.

\begin{lem}\label{lem:inclusion} 
	Let the self-adjoint operator $\OpIId$ be as in
	Definition~\ref{dfn:operator}. 
	Then for any $u \in \cD(\RcoLambda)$ holds 
	\begin{equation}\label{u}
		u\in \dom\OpIId\quad\text{and} \quad 	\OpIId u = -\Delta u.
	\end{equation}
\end{lem}
\begin{proof}
	Let $\Sigma$ and $\Omega_\pm$ be as in Hypothesis~\ref{hyp:curve}.
	Let $u\in\cD(\RcoLambda)\subset F_\Lambda\subset \dom\frmIId$
	and $v\in \dom \frmIId$. Define $u_\pm := u\uhr\Omega_\pm$
	and $v_\pm := v\uhr\Omega_\pm$. With these notations in hands we get
	\[
		\frmIId [u,v] 
		=
		(\nabla u_+,\nabla v_+)_{\Omega_+} +
		(\nabla u_-,\nabla v_-)_{\Omega_-},
	\]
	where the boundary term in~\eqref{eq:form} 
	vanished due to the choice of $u$.
	Applying the first Green identity (see \eg \cite[Lem. 4.1]{McL}
	and also \cite[Sec. 2]{BEL14_RMP}) to the above formula, we get
	\[
	\begin{split}
		\frmIId[u,v] & = (-\Delta u_+, v_+)_{\Omega_+}
					+
					(-\Delta u_-, v_-)_{\Omega_-}\\
					&\qquad\qquad+
					\big(\p_{\nu_+} u_+|_\coLambda +  \p_{\nu_-} u_-|_\coLambda,
							 v|_\coLambda\big)_\coLambda\\
					&\qquad\qquad\qquad\qquad 
					+ 
					\big(\p_{\nu_+} u_+|_\Lambda, v_+|_\Lambda\big)_\Lambda
					+
					\big(\p_{\nu_-} u_-|_\Lambda, v_-|_\Lambda\big)_\Lambda
				= (-\Delta u,v)_{\dR^2},
	\end{split}
	\]
	where we employed that 
	$\p_{\nu_\pm}u_\pm|_\Lambda = 0$, that 
	$\p_{\nu_+} u_+|_\coLambda + \p_{\nu_-}u_-|_\coLambda = 0$,
	and that $[v]_\coLambda = 0$; 
	for the latter \cf Proposition~\ref{prop:trace} (i).
	Finally, the first representation theorem 
	 yields~\eqref{u}.
\end{proof}

Next, we define the notion of the quasi-conical domain;
see \cite{Gl} and also \cite[Def. X.6.1]{EE}.

\begin{dfn}\label{def:quasi_conical}
	A domain $\Omega\subset\dR^2$ is called \emph{quasi-conical}
	if for any $n\in\dN$ there exists $x_n\in\dR^2$
	such that $D_n(x_n) \subset\Omega$. Recall
	that here $D_n(x_n)$ is the disc of radius $n$
	with the center $x_n$.
\end{dfn}

Using this notion, we prove that positive semi-axis
lies inside the spectrum of $\OpIId$ if the domain
$\RcoLambda$ is quasi-conical. The technique of this
proof is rather standard.

\begin{prop} 
\label{prop:ess1}
	Let the curve $\Lambda\subset\dR^2$ as in Hypotheses~\ref{hyp:curve}
	be such that the domain $\RcoLambda$ is quasi-conical.
	Then the spectrum of the self-adjoint operator 
	$\OpIId$ in Definition~\ref{dfn:operator} satisfies 
	\begin{equation}\label{sess:incl}
		\s(\OpIId)\supseteq [0,+\infty).
	\end{equation}
\end{prop}

\begin{proof}
	First, for any $k\in\dR^2$, define the sequence
	\[
		u_n(x) := v_n(x)\ee^{\ii k\cdot x},	\qquad n\in\dN,
	\]
	where $v_n(x) := n^{-1}v(n^{-1}x)$, $n\in\dN$,
	and $v$ is a non-trivial function in $\cD(\dR^2)$
	with $\supp v \subset D_1$ and such that $\|v\|_{\dR^2} = 1$.
	The prefactor in the definition of $v_n$ is chosen
	in such a way that also each $v_n$ satisfies $\|v_n\|_{\dR^2} = 1$.
	In fact, we have (by direct computations)
	\begin{equation}\label{eq:wt_vn}
		\|v_n\|_{\dR^2} = 1, 
		\qquad
		\|\nabla  v_n\|_{\dR^2} = \frac{\|\nabla v\|_{\dR^2}}{n},
		\qquad
		\|\Delta v_n\|_{\dR^2} = \frac{\|\Delta v\|_{\dR^2}}{n^2}.
	\end{equation}

	Secondly, we set
	\[
		w_n(x) := u_n(x - x_n), \qquad n\in\dN,
	\]
	with $x_n$ corresponding to the quasi-conical domain $\RcoLambda$
	according to Definition~\ref{def:quasi_conical}.
	Hence, we get
	\[
		\supp w_n \subset  D_n(x_n) \subset \RcoLambda,
	\]
	and therefore $w_n \in \cD(\RcoLambda)$	for all $n\in\dN$.
	It is clear in view of Lemma~\ref{lem:inclusion}
	that each $w_n$ belongs to 
	$\dom \OpIId \supset \cD(\RcoLambda)$.
	
	A direct computation yields
	\[
		|-\Delta w_n - |k|^2 w_n|^2 
				\le 2|\Delta v_n|^2 + 4|k\cdot\nabla v_n|^2
				\le 2|\Delta v_n|^2 + 4|k|^2\cdot|\nabla v_n|^2.
	\]
	Using \eqref{eq:wt_vn} and Lemma~\ref{lem:inclusion}, we therefore have
	\[
	\begin{split}
		\|\OpIId w_n - |k|^2 w_n\big\|^2_{\dR^2}
		 =
		\| -\Delta w_n - |k|^2 w_n\|^2_{\dR^2}
		 \le 
		2\|\Delta v_n\|^2_{\dR^2} + 4|k|^2 \|\nabla v_n\|^2_{\dR^2}
		\arr 0,\!\qquad n\arr \infty.
	\end{split}	
	\]
	Since the choice of $k\in\dR^2$ was arbitrary,
	we conclude applying Weyl's criterion (see \cite[Sec. 7.4]{W} 
	and also \cite[Thm. 4]{KLu14})	that $[0,+\infty)\subseteq\s(\OpIId)$.
\end{proof}

We emphasize that not for every non-closed curve $\Lambda\subset\dR^2$ the domain $\RcoLambda$
is quasi-conical; \eg  for the \emph{Archimedean spiral},  defined in polar
coordinates $(r,\phi)$ by the equation $r(\varphi) := a + b\varphi$, 
$\varphi\in\dR_+$, $a, b > 0$, the domain $\RcoLambda$ is not of
this type.

In the case of bounded curves we show that 
the essential spectrum of $\OpIId$ coincides 
with the set $[0,+\infty)$.

\begin{prop}\label{prop:ess2}
	Let the bounded curve $\Lambda\subset\dR^2$ be as in 
	Hypothesis~\ref{hyp:curve} and let the self-adjoint operator 
	$\OpIId$ be as in Definition~\ref{dfn:operator}. 
	Then its essential spectrum is characterised as 
	\[
		\sess(\OpIId) = [0,+\infty).
	\]
\end{prop}

\begin{proof}
	Let the curve $\Sigma\subset\dR^2$ and the domains
	$\Omega_\pm\subset\dR^2$ be as in Hypothesis~\ref{hyp:curve}, 
	in particular, $\Lambda\subset\Sigma$. 
	Let us also set  $c := \|\omega\|_\infty$.
	Consider the sesquilinear form
	\begin{equation}\label{frac}
	\begin{split}
		\fra^\Sigma_c[u,v] & := 
		(\nabla u,\nabla v)_{\dR^2} -c ( [u]_\Sigma, [v]_\Sigma )_\Sigma,\\
		\dom\fra^\Sigma_c & := H^1(\Omega_+)\oplus H^1(\Omega_-).
	\end{split}	
	\end{equation}
	According to~\cite[Prop. 3.1]{BEL14_RMP} the form
	$\fra^\Sigma_c$ is closed, densely defined, symmetric, 
	and lower-semibounded in $L^2(\dR^2)$. 
	The self-adjoint operator $\sfH_c^\Sigma$
	in $L^2(\dR^2)$ representing the form $\fra^\Sigma_c$,
	satisfies 
	\begin{equation}\label{ess1}
		\sess(\sfH^\Sigma_c) = [0,+\infty);
	\end{equation}
	see~\cite[Thm. 4.2]{BEL14_RMP} and also~\cite[Thm. 3.16]{BLL13_AHP}.
	The sesquilinear forms $\frmIId$ and $\fra_c^\Sigma$ 
	in~\eqref{eq:form} and~\eqref{frac}, respectively, 
	naturally satisfy the ordering
	\[
		\fra_c^\Sigma\prec\frmIId
	\]
	in the sense of~\cite[\S VI.2.5]{Kato},
	see also \cite[\S 10.2.3]{BS87}. Indeed,
	firstly, $\dom\frmIId \subset \dom \fra_c^\Sigma$
	and, secondly, for any $u \in \dom\frmIId$ 
	the inequality $\fra_c^\Sigma[u,u]\le\frmIId[u,u]$ 
	holds due to the choice of the constant $c \ge 0$.
	Hence, using~\eqref{ess1} and~\cite[\S 10.2, Thm. 4]{BS87}
	we arrive at
	\[
		0 = \inf\sess(\sfH_c^\Sigma) \le\inf\sess(\OpIId).
	\]
	Therefore, we end up with the following inclusion
	\begin{equation}\label{ess2}
		\sess(\OpIId) \subseteq [0,+\infty).
	\end{equation}
	Moreover, for simple geometric reasons for any bounded curve $\Lambda$
	the domain $\RcoLambda$ is quasi-conical
	and hence by Proposition~\ref{prop:ess1} the opposite inclusion
	\begin{equation}\label{ess3}
		\sess(\OpIId) \supseteq [0,+\infty)
	\end{equation}
	holds as well. 
	The claim then follows from these two inclusions 
	(\eqref{ess2} and \eqref{ess3}).
\end{proof}

\section{Non-negativity of $\OpIId$}
\label{sec:abs}

This section plays the central role in the present paper. 
We obtain various sufficient conditions for the operator 
$\OpIId$ to be non-negative. Under additional assumptions we 
also show that positive spectrum of $\OpIId$ comprises
the whole positive real axis and thus the operator 
$\OpIId$ has no bound states.
In the proofs we use the min-max principle
for self-adjoint operators, a reduction to the one-dimensional model problem
discussed in Subsection~\ref{ssec:modelproblem},  
and some insights from geometry and complex analysis. 

\subsection{An auxiliary lemma}

In this subsection we prove a lemma, based on which we show non-negativity
of the operators $\OpIId$ under certain assumptions on $\omega$. 
For the formulation of this lemma we require the following hypothesis,
the assumptions of which  are grouped
in three logical blocks labelled by capital latin letters.

\begin{hyp}\label{hyp:disc}
	{\rm 
	{\bf (A)}
	Let a monotone piecewise-$C^1$ curve $\Lambda\subset\dR^2$ be parametrized
	via the mapping $\phi\colon (0,R) \arr \dR$, $R \in (0,+\infty]$, 
	as in Definition~\ref{def:monotone} with $\xx_0 = 0$.

	{\bf (B)}
	Suppose that piecewise-$C^1$ domains 
	$G_\pm \subset D_R$ satisfy the following conditions:
	\[
		G_+\cap G_- = \varnothing, \quad \ov{D_R} = \ov{G_+\cup G_-},
		\quad \text{and}\quad 
		\Lambda \subset\ov{G_+}\cap\ov{G_-}.
	\]	
	Set $\Sigma := \ov{G_+}\cap\ov{G_-}$.
	In particular, the inclusion $\Lambda\subset \Sigma$ holds. 
	
	{\bf (C)}
	Let the function $\omega \in L^\infty(\Lambda;\dR)$ 
	as a function of the distance $r$ from the origin 
	satisfy 
	\begin{equation}\label{omega}
		\omega(r) \le 
		\frac{1}{2\pi r\sqrt{1+(r\varphi^\pp(r))^2}}, 
		\qquad\text{for}\quad r\in(0,R).
	\end{equation}
	}
\end{hyp}

We further deal with the space $H^1(G_+)\oplus H^1(G_-) \subset L^2(D_R)$.
Let us introduce also the following notations
\[
	[u]_\bullet := u_+|_\bullet - u_-|_\bullet,
	\quad
	\bullet = \Lambda~\text{or}~\coLambda,
	\quad
	u = u_+\oplus u_-\in H^1(G_+)\oplus H^1(G_-).
\]	
Clearly, one can define polar coordinates $(r,\varphi)$ on $D_R$, which 
are connected with the usual Cartesian coordinates via
standard relations $x = r\cos\varphi$ and $y = r\sin\varphi$. 
The points $(r,\varphi + 2\pi k)$ with $k\in\dZ$ are identified
with each other. The disc $D_R$ in the polar coordinate system
is given by $D_R = \big\{(r,\varphi)\colon r\in [0,R),~ 
\varphi \in [0,2\pi)\big\}$.
	
For the substantial simplification of further
computations we make use of the following
shorthand notation:
\begin{equation}\label{frf}
	\frfR[ u ] := \| \nabla u \|^2_{D_R} - 
					 (\omega [u]_\Lambda, [u]_\Lambda)_\Lambda,
	\qquad
	u\in \cD(\ov{G_+})\oplus \cD(\ov{G_-}),
\end{equation}
where all the objects are as in Hypothesis~\ref{hyp:disc}.
Now we formulate and prove the following lemma.

\begin{lem}\label{lem:estimate}
	Assume that Hypothesis~\ref{hyp:disc} holds. 
	Then $\frfR [ u ] \ge  0$ for all 
	$u\in \cD(\ov{G_+})\oplus\cD(\ov{G_-})$ such that
	$[u]_\coLambda = 0$.
\end{lem}

\begin{proof}
	Let $u\in \cD(\ov{G_+})\oplus\cD(\ov{G_-})$ be such that 
	$[u]_\coLambda = 0$. The proof of $\frfR [ u ] \ge  0$
	is then split in three steps.
	
	\emph{Step 1.} 
	For any $(x,y)\in D_R\setminus \Sigma$ the value $|(\nabla u)(x,y)|^2$ 
	can be expressed in polar coordinates $(r,\varphi)$ as
	\[
		|(\nabla u)(x,y)|^2  = 
			|(\p_r u)(r,\varphi)|^2  
						+\frac{1}{r^2}|(\p_\varphi u)(r, \varphi)|^2.
	\]
	Using the above expression for the gradient we obtain the following 		
	estimate
	\begin{equation}\label{grad}
	\begin{split}
		\|\nabla u\|_{D_R}^2  
		& = 
		\int_0^{2\pi}\int_0^R |(\nabla u)(r,\varphi)|^2r\dd r \dd \varphi
		\le
		\int_0^R\frac{1}{r} \bigg(\int_0^{2\pi}  
			|(\p_\varphi u)(r,\varphi)|^2\dd \varphi\bigg) \dd r,
	\end{split}		
	\end{equation}
	in which we have thrown away a positive term in the second step. 
	Interchanging of the integrals
	in the above computation can be justified by Fubini's theorem 
	(see \eg \cite[Ch. 2, Thm. 3.1]{StSh}).
	
	\emph{Step 2.}
	Using the mapping $\phi \colon (0,R)\rightarrow \dR$ 
	as in Hypothesis~\ref{hyp:disc}~(A) we 
	define the following auxiliary function
	\[
		\frj(r) := \sqrt{1 + (r\varphi^\pp(r))^2},\qquad r\in(0,R).
	\]
	The curvilinear integral along $\Lambda$ in~\eqref{frf}
	can be rewritten in terms of the mapping 
	$\phi\colon (0,R)\arr \dR$ and the function $\frj$
	in the conventional way 
	and then further estimated with the help of assumption~\eqref{omega} 
	\begin{equation}\label{perturb} 
	\begin{split}
		(\omega [u]_\Lambda, [u]_\Lambda)_\Lambda
		 & = 
		 \int_0^R \omega(r) \frj(r) |u_+(r,\phi(r)) - u_-(r,\phi(r))|^2  \dd r\\
		&\le \int_0^R \frac{1}{2\pi r}\big|u_+(r,\phi(r)) - u_-(r,\phi(r))\big|^2 \dd r.
	\end{split}
	\end{equation}

	\emph{Step 3.}
	Define the following function
	\begin{equation}\label{Sr}
		S(r) := \int_0^{2\pi}|(\p_\varphi u)(r,\varphi)|^2 
									\dd \varphi 
				- 
			\frac{1}{2\pi}\big|u_+(r,\varphi(r)) - u_-(r,\varphi(r))\big|^2,
	\end{equation}
	where $r\in (0,R)$.	
	Thanks to the choice of $u$, for all $r\in (0,R)$ 
	the function $[0,2\pi) \ni \phi \mapsto u(r,\varphi)$
	can naturally be identified with the
	piecewise-$C^1$ function $\psi_r$ on the interval $I = (0,2\pi)$,
	%
	which by \cite[App. E]{S} belongs to $H^1(I)$.
	Moreover, the relation $S(r) = \fra_{d,\omega}[\psi_r]$
	holds with the form $\fra_{d,\omega}$ as in \eqref{pointform}, 
	where $d = 2\pi$ and $\omega = 1/2\pi$. In particular, 
	$d\omega =  2\pi /2\pi =  1$
	and by Corollary~\ref{cor:nonnegative} we obtain
	\[
		S(r) \ge 0,\qquad\text{for all}\quad r\in (0,R).
	\]
	Finally, using \eqref{grad}, \eqref{perturb} 
	and non-negativity of $S(r)$ we arrive at
	\[
		\frfR [u] \ge  \int_0^R \frac{S(r)}{r} \dd r \ge 0.
	\]
\end{proof}

\subsection{Non-negativity of $\OpIId$ for monotone $\Lambda$}
\label{ssec:abs_main}

In this subsection we obtain various
explicit sufficient conditions on $\omega$ ensuring non-negativity 
of $\OpIId$ assuming that $\Lambda$ is monotone.  
General results are illustrated with
two examples: an Archimedean spiral 
and a subinterval of the straight line in $\dR^2$. 

\begin{thm}\label{thm:absence2}
	Let a monotone piecewise-$C^1$ curve 
	$\Lambda\subset \dR^2$ be 
	parametrized via $\phi\colon (0,R) \arr \dR$, $R\in (0,+\infty]$,
	as in Definition~\ref{def:monotone}.
	Let the self-adjoint operator $\OpIId$ 
	be as in Definition~\ref{dfn:operator}
	with $\omega\in L^\infty(\Lambda;\dR)$.
	Then 
	\[
		\s(\OpIId)\subseteq [0,+\infty)
		\qquad \text{if}\quad
		\omega(r) \le \frac{1}{2\pi r\sqrt{1+ (r\phi^\pp(r))^2}},
		\quad\text{for}
		\quad r\in(0,R).
	\]
	If $\omega$ is majorized as above, and additionally, the domain
	$\RcoLambda$ is quasi-conical, then $\s(\OpIId) = [0,+\infty)$.
\end{thm}

\begin{proof}
	Let $\Sigma$ and $\Omega_\pm$ be as in Hypothesis~\ref{hyp:curve}.
	Without loss of generality we assume that $\xx_0 = 0$
	in Definition~\ref{def:monotone}.
	
	Let us define the complement 
	$\Omega_{\rm c} := \dR^2\setminus\ov{D_R}$ of the disc $D_R$,
	the curve $\Gamma := \Sigma\cap D_R$, and
	the domains $G_\pm := \Omega_\pm\cap D_R$.
	It is straightforward to see that the tuple 
	$\{D_R, G_+, G_-, \Lambda,\omega\}$ satisfies Hypothesis~\ref{hyp:disc}.
	
	Let $u \in F_\Lambda$ and define $u_R := u\uhr D_R$,
	$u_{\rm c} := u\uhr\Omega_{\rm c}$.
	Then it holds that
	\[
		u_R \in \cD(\ov{G_+})\oplus\cD(\ov{G_-}) 
		\qquad\text{and}\qquad
		[u_R]_{\Gamma\setminus\Lambda} = 0.
	\]
	Hence, using Lemma~\ref{lem:estimate} we get
	\[
		\frmIId [u,u] 
		= 
		\frfR [u_R] + \|\nabla u_{\rm c}\|^2_{\Omega_{\rm c}}
		\ge 
		\frfR [u_R] \ge 0.
	\]
	Since $F_\Lambda$ is a core for the form $\frmIId$,
	we get by \cite[Thm. 4.5.3]{D95} that
	the self-adjoint operator $\OpIId$ is non-negative.
	If, additionally, the domain $\RcoLambda$ is quasi-conical, 
	Proposition~\ref{prop:ess2} implies that
	\[
		\s(\OpIId) = [0,+\infty).
	\]
\end{proof}

\begin{example}\label{example2}
	Let the piecewise-$C^1$ curve $\Lambda\subset\dR^2$ be defined as
	\[
		\Lambda := \{(r\cos(r), r\sin(r))\in\dR^2\colon r\in\dR_+\}.
	\]
	Obviously, this curve is 
	monotone in the sense of Definition~\ref{def:monotone}
	with $\xx_0 = 0$ and $\phi(r) := r$, $r\in (0,+\infty)$.
	The curve $\Lambda$ is a special case of an Archimedean spiral.
	Theorem~\ref{thm:absence2} yields that 
	\[
		\s(\OpIId)\subseteq [0,+\infty)
		\qquad \text{if}\quad
		\omega(r) \le \frac{1}{2\pi r\sqrt{1+r^2}},
		\quad\text{for}\quad r > 0.
	\]
	%
\end{example}

The case of a non-varying interaction strength $\omega$ is 
of special interest. In the rest of this subsection we assume
for the sake of demonstrativeness that $\omega\in\dR$ is a constant.
Define also the following characteristic of a bounded monotone
piecewise-$C^1$
curve $\Lambda\subset\dR^2$ 
(parametrized as in Definition~\ref{def:monotone})
\begin{equation}\label{omega*}
	\omega_*(\Lambda) := \inf_{r\in(0,R)}
						\frac{1}{2\pi r\sqrt{1 + (r\varphi^\pp(r))^2}}.
\end{equation}					
It is not difficult to see that $0 < \omega_*(\Lambda) <+\infty$.

The following corollary is a direct consequence of Theorem~\ref{thm:absence2},
Proposition~\ref{prop:ess2}, and simple geometric argumentation.

\begin{cor}\label{cor:absence}
	Let $\Lambda\subset \dR^2$  be a bounded monotone piecewise-$C^1$ curve 
	and let the self-adjoint operator $\OpIId$ 
	be as in Definition~\ref{dfn:operator} with non-varying strength 
	$\omega\in\dR$.
	Then 
	\[
		\s(\OpIId)  = [0,+\infty)
		\qquad\text{for all}~
		\omega \in (-\infty, \omega_*]
	\]  
	with $\omega_* = \omega_*(\Lambda) > 0 $ defined in~\eqref{omega*}. 
\end{cor}

To illustrate this corollary we provide an example.

\begin{example}\label{example1}
	Consider the interval of length $L > 0$ in the plane: 
	\begin{equation}\label{LambdaL}
		\Lambda := \{(x,0)\in\dR^2\colon 0 < x < L\}.
	\end{equation}
	Clearly, the interval $\Lambda$
	is monotone in the sense of Definition~\ref{def:monotone}
	with $\xx_0 = 0$ and $\phi(r) = 0$, $r\in (0,L)$.
	Then we get from Corollary~\ref{cor:absence}, using 
	formula~\eqref{omega*}, that
	\[
	 	\s(\OpIId) = [0,+\infty)
	 	\qquad
		\text{for all}~\omega \in \big(-\infty,\tfrac{1}{2\pi L}\big].
	\]
\end{example}

\begin{remark}
	Let $\Lambda$ be as in \eqref{LambdaL}.
	It is worth noting that the result of the above example can
	be improved in the following way. Define the points
	$\xx_0 = (0,0)$, $\xx_1 = (0,L)$, the intervals
	\[
		\Lambda_0 := \{(x,0)\in\dR^2\colon 0 < x < L/2\},
		\qquad
		\Lambda_1 := \{(x,0)\in\dR^2\colon L/2 < x < L\},
	\]
	the discs $D_{L/2}(\xx_0)$, $D_{L/2}(\xx_1)$, 
	and the complement 
	\[	
		\Omega_{\rm c} 
		:= 
		\dR^2\setminus(\ov{D_{L/2}(\xx_0)\cup D_{L/2}(\xx_1)})
	\]
	of the closure of their union. 
	Let $u\in F_\Lambda$ and define $u_k := u\uhr D_{L/2}(\xx_k)$, 
	$k=0,1$, $u_{\rm c} := u\uhr \Omega_{\rm c}$. 
	Assuming that $\omega \in (-\infty, \frac{1}{\pi L}]$,
	we get by Lemma~\ref{lem:estimate} that
	\[
		\frmIId[u,u] = 
		\frf^{\Lambda_0}_{D_{L/2}(\xx_0),\omega}[u_0] 
		+
		\frf^{\Lambda_1}_{D_{L/2}(\xx_1),\omega}[u_1] 
		+ 
		\|\nabla u_{\rm c}\|^2_{\Omega_{\rm c}} \ge 0.
	\] 
	Thus, the operator $\OpIId$ is non-negative and by 
	Proposition~\ref{prop:ess2} we get $\s(\OpIId) = [0,+\infty)$.
\end{remark}

One may expect that for a sufficiently large coupling 
constant $\omega > 0$ or for a sufficiently long curve
$\Lambda$ negative spectrum of the self-adjoint operator
$\OpIId$ is non-empty. In the next
proposition we confirm this expectation via an example.

\begin{prop}\label{prop:interval}
	Let $\Lambda\subset\dR^2$ be as in~\eqref{LambdaL} 
	and the self-adjoint operator
	$\OpIId$ be as in Definition~\ref{dfn:operator}
	with non-varying strength $\omega\in\dR$.
	Then 
	\[
		\sd(\OpIId)\cap (-\infty,0) \ne \varnothing
		\qquad \text{for all}~
		\omega \in \big(\tfrac{\pi}{2L}, +\infty\big).
	\]
\end{prop}

\begin{proof}
	Let us split the plane $\dR^2$ into three domains
	\[
		\Omega_1 := (-\infty,0)\times \dR,
		\qquad 
		\Omega_2 := (0,L)\times \dR,
		\qquad 
		\Omega_3 := (L,+\infty)\times\dR,
	\]
	via straight lines 
	\[
		\Pi_1 := \{0\}\times\dR, \qquad \Pi_2 := \{L\}\times\dR,
	\]
	as indicated in Figure~\ref{fig}.

	\begin{figure}[h]
		 \includegraphics[width=0.5\textwidth]{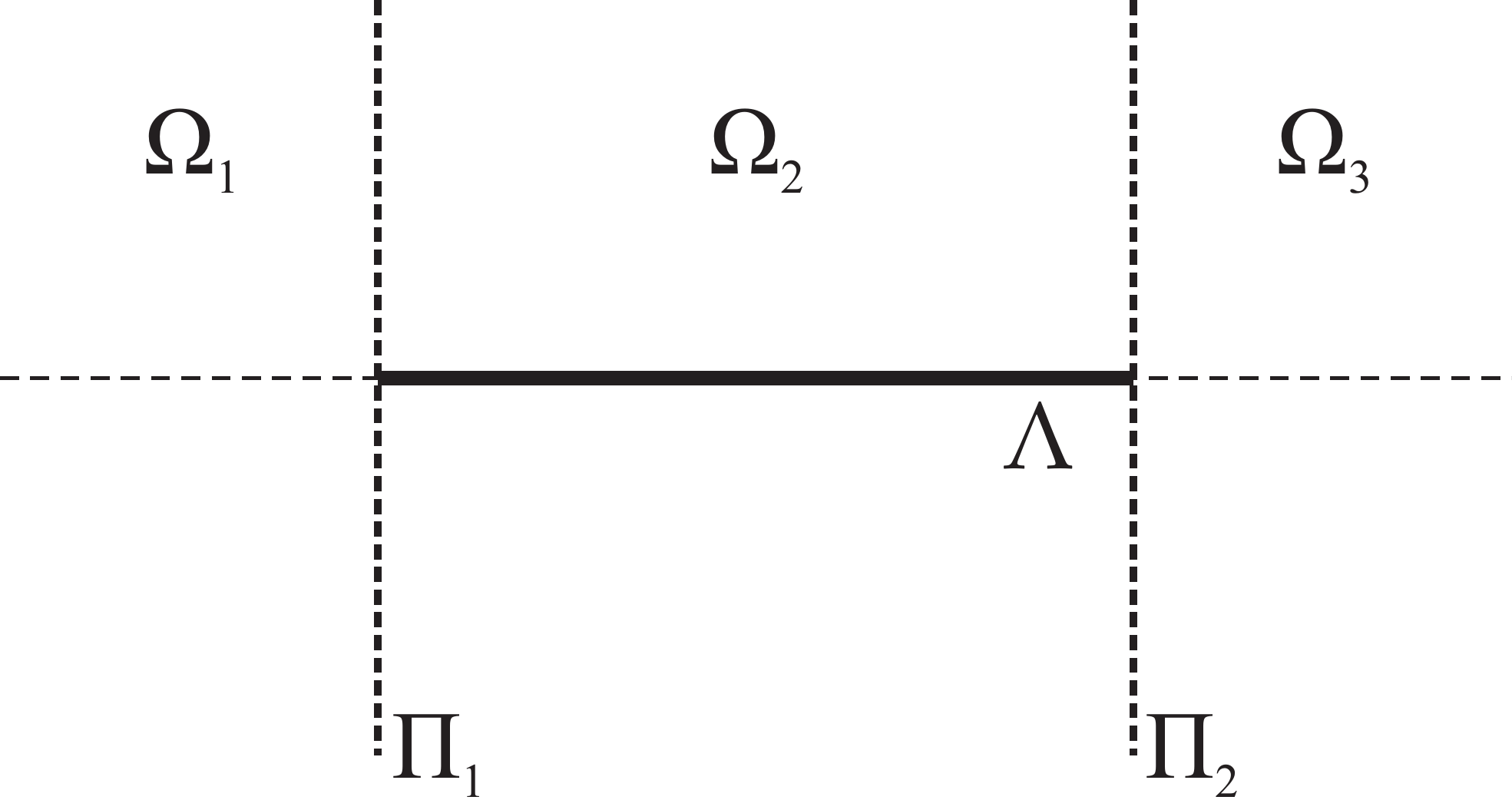}
		 \\
		 \caption{Splitting of $\dR^2$ into three domains 
		 $\{\Omega_k\}_{k=1}^3$.}\label{fig}
	\end{figure}
	
	Consider the sesquilinear form
	\[
		\fra_{\omega, \rm D}^\Lambda[u,v] := \frmIId[u,v],
		\qquad 
		\dom \fra_{\omega,\rm D}^\Lambda 
		:= 
		\big\{u\in \dom\frmIId\colon u|_{\Pi_k} = 0,~k=1,2\big\}.
	\]
	It is not difficult to check that the sesquilinear form
	$\fra_{\omega, \rm D}^\Lambda$ is closed, symmetric, densely
	defined, and semibounded in $L^2(\dR^2)$. This form induces
	via the first representation theorem the self-adjoint operator
	$\sfH_{\omega, \rm D}^\Lambda$ in $L^2(\dR^2)$, which can 
	be represented as the orthogonal sum
	$\sfH_1\oplus \sfH_2\oplus \sfH_3$ with respect to the decomposition
	$L^2(\dR^2) = \oplus_{k=1}^3 L^2(\Omega_k)$. 
	Note that $\sfH_1$ and $\sfH_3$ are both non-negative 
	and their spectra are
	given by the set $[0,+\infty)$. 
	The spectrum of $\sfH_2$ can be computed via separation of variables in the 
	strip $\Omega_2$. In particular, 
	the ground state of $\sfH_2$ corresponds to the eigenvalue
	\[
		\lambda_1(\sfH_2) = \frac{\pi^2}{L^2} -4\omega^2,
	\]
	where we used that the one-dimensional Schr\"odinger operator on the
	full-line with one-center point $\delta^\pp$-interaction of strength 
	$\omega > 0$
	has the lowest
	eigenvalue $-4\omega^2$; \cf \cite[Ch. I.4]{AGHH05},
	where not $\omega$,
	but $\beta = 1/\omega$ is called the strength of $\delta^\pp$-interaction. 
 
	If the assumption in the formulation of the proposition holds, 
	then $\lambda_1(\sfH_2) < 0$ and the operator 
	$\sfH_{\omega,\rm D}^{\Lambda}$ has at least one negative eigenvalue. 

	It remains to note that by Proposition~\ref{prop:ess2} we have
	$\sess(\OpIId) = [0,+\infty)$ and that the form ordering
	\[
		\frmIId \prec\fra_{\omega,\rm D}^\Lambda
	\]
	can easily be verified, which yields by \cite[\S 10.2, Thm. 4]{BS87} 
	that the operator $\OpIId$ has at least one negative eigenvalue.
\end{proof}

\subsection{Absence of bound states for $\OpIId$ and LFTs}
\label{ssec:abs_LFT}

In this subsection we show absence of bound states in the 
weak coupling regime for a class of bounded non-monotone piecewise-$C^1$ curves, which are 
(with minor restrictions) images of bounded monotone
curves under LFTs. Since the identical transform $M(z) = z$ is an LFT, 
this class is certainly
larger than the class of bounded monotone curves.
As an example we treat $\delta^\pp$-interaction supported
on a circular arc subtending an angle $\tt > \pi$.

First, we provide for convenience of the reader
two standard claims on change of variables under LFT.
The proofs of them are outsourced to Appendix~\ref{app:A}.

\begin{lem}\label{lem:change1}
	Let $\Lambda\subset \dR^2$ be 
	a bounded curve as in Hypothesis~\ref{hyp:curve},
	let the space $F_\Lambda$ be as in~\eqref{FLambda}, and let 
	$M\colon \wh\dC\arr\wh\dC$ be an LFT
	such that $M(\infty), M^{-1}(\infty) \notin \Lambda$.
	Then for any $u \in F_\Lambda$ 
	\begin{equation}\label{nablauv}
		\int_{\dR^2} |\nabla u|^2 \dd x  
		= 
		\int_{\dR^2} |\nabla v|^2 \dd x
	\end{equation}
	holds with $v := u \circ M$. 
\end{lem}

\begin{remark}
	The function $v$ itself
	in the formulation of the above lemma is continuous and piecewise smooth,
	but it is not necessarily compactly supported or square-integrable.
\end{remark}

\begin{lem}\label{lem:change2}
	Let $\Lambda\subset\dR^2$ be a bounded piecewise-$C^1$ curve, 
	parametrized via the mapping $\lm\colon I \arr \dR^2$, $I := (0,L)$,
	as in Definition~\ref{def:C1},
	let the space $F_\Lambda$ be as in~\eqref{FLambda}
	and let $\omega\in L^\infty(\Lambda;\dR)$.
	For an LFT $M\colon\wh\dC\arr\wh\dC$ with $\J$ as in \eqref{J} 
	and such that $M(\infty), M^{-1}(\infty) \notin \Lambda$ define
	$\Gamma := M^{-1}(\Lambda)$, $\gamma := M^{-1} \circ \lm$
	and
	\begin{equation}\label{def:wt_omega}
		\wt\omega( \gamma( s )) := 	\omega( \lm(s)) \sqrt{\J( \gamma(s))},
		\qquad s\in I.
	\end{equation}			
	Then  the relation
	\[
		\big(\omega [u]_\Lambda, [u]_\Lambda\big)_\Lambda 	
			= 	
		\big(\wt\omega [v]_\Gamma, [v]_\Gamma\big)_\Gamma
	\]
	holds for any $u \in  F_\Lambda$ and $v := u \circ M$.	
\end{lem}

\begin{remark}
	Note that the function $v$ in the formulation of the above lemma
	does not belong to $F_\Gamma$
	in general. However, $v_\pm := v\uhr M^{-1}(\Omega_\pm)$
	with $\Omega_\pm$ as in Hypothesis~\ref{hyp:curve}
	are well-defined and continuous up to $\Gamma$.
	Hence, the restrictions $v_\pm|_\Gamma$
	are meaningful 
	and $[v]_\Gamma := v_+|_\Gamma - v_-|_\Gamma$ is well-defined.
\end{remark}

Now we can formulate the key result of this subsection, whose proof
with all the above preparations is rather short.

\begin{thm}\label{thm:absenceLFT}
	Let $\Lambda \subset \dR^2$ be a bounded piecewise-$C^1$ curve
	and let the self-adjoint operator $\OpIId$
	in $L^2(\dR^2)$ be as in Definition~\ref{dfn:operator}
	with non-varying strength $\omega\in\dR$.
	Suppose that there exists an LFT $M\colon \wh\dC\arr\wh\dC$ 
	such that:
	\begin{myenum}	
		\item [\rm (a)] $M(\infty), M^{-1}(\infty) \notin \Lambda$;
		\item [\rm (b)] $\Gamma := M^{-1}( \Lambda )$ is monotone.
	\end{myenum}
	Let the constant $\omega_*(\Gamma) > 0$ 
	be associated to $\Gamma$ via \eqref{omega*}.		
	Then it holds that
	\[
		\s(\OpIId) = [0,+\infty)
		\qquad
		\text{for all}
		~
		\omega \in (-\infty,\omega_*],
	\]	
	where 
	$\omega_* :=  \frac{\omega_*(\Gamma)}{\sup_{\zz\in\Gamma}\sqrt{\J(\zz)}}$.
\end{thm}
\begin{proof}
	\emph{Step 1.}
	Suppose that the curve $\Lambda$ is parametrized via the mapping
	$\lm\colon (0,L)\arr \dR^2$ 
	(as  in Definition~\ref{def:C1}).	
	Define the mapping $\gamma := M^{-1} \circ \lambda$. 
	Due to assumption (a)
	the curve $\Gamma$ is bounded and the mapping $\gamma$
	parametrizes it.
	Without loss of generality suppose that 
	the curve $\Gamma$ is monotone in the sense of 
	Definition~\ref{def:monotone} with $\xx_0 = \gamma(0) = 0$
	and with $\phi\colon (0,R)\arr \dR$, $R = |\gamma(L)|$.
	Consider the complement 
	$\Omega_{\rm c} := \dR^2\setminus\ov{D_R}$
	of the disc $D_R$. Let the curve $\Sigma$ and the domains
	$\Omega_\pm$ be associated to $\Lambda$  
	as in Hypothesis~\ref{hyp:curve}.
	
	Define auxiliary domains $G_\pm := M^{-1}(\Omega_\pm)\cap D_R$. 	
	Thus, the splitting
	\[	
		D_R = G_+\, \dot\cup \, M^{-1}(\Sigma)\, \dot\cup \, G_-
	\] 
	holds. Let $\wt\omega$ be defined 
	via the formula~\eqref{def:wt_omega} in Lemma~\ref{lem:change2}.
	Hence, we obtain
	\[
		\wt\omega 
		\le 
		\omega \sup_{z\in\Gamma}\sqrt{\J(z)}
		\leq
		\omega_*\sup_{z\in\Gamma}\sqrt{\J(z)}
		= \omega_*(\Gamma).
	\]
	Summarizing, the tuple $\{ D_R, G_+, G_-, \Gamma, \wt\omega\}$
	fulfils Hypothesis~\ref{hyp:disc}. 
	
	\emph{Step 2.}
	Let $u \in F_\Lambda$ with $F_\Lambda$ as in~\eqref{FLambda}
	and define the composition $v := u\circ M$. 
	Set $v_R := v\uhr D_R$ and $v_{\rm c} := v\uhr \Omega_{\rm c}$.
	Using Lemmas~\ref{lem:change1} and~\ref{lem:change2} 
	we obtain
	\begin{equation}\label{eq:equality_of_forms}
	\begin{split}
		\frmIId[ u, u ] 
		& = 
		\|\nabla u\|^2_{\dR^2} - (\omega [u]_\Lambda, [u]_\Lambda)_\Lambda
		=
		\|\nabla v\|^2_{\dR^2} - (\wt\omega [v]_\Gamma, [v]_\Gamma)_\Gamma \\
		& =
		\| \nabla v_R \|^2_{D_R} - 
				(\wt\omega [v_R]_\Gamma, [v_R]_\Gamma)_\Gamma
		+ \|\nabla v_{\rm c}\|^2_{\Omega_{\rm c}}
		\ge 
		\frf^\Gamma_{D_R,\wt\omega}[ v_R, v_R ] \ge 0,		
	\end{split}	
	\end{equation}
	where we applied Lemma~\ref{lem:estimate} in the last step.
	Hence, the operator $\OpIId$ is non-negative.
	
	\emph{Step 3.}
	Since the curve $\Lambda$ is bounded, 
	Proposition~\ref{prop:ess2} applies, and we arrive at
	$\sess(\OpIId) = [0,+\infty)$. The results of Step 2 and Step 3
	imply the claim.
\end{proof}

To conclude this subsection we show that a  model of sufficiently weak 
$\delta^\pp$-interaction of non-varying strength
supported on a circular arc 
subtending the angle $2\pi - 2\eps$ ($\eps\in (0,\pi)$)
has no bound states in the weak coupling regime.
We emphasize that circular arcs subtending angles $\tt > \pi$
are non-monotone and the results of the previous subsection
do not apply to them.
\begin{figure}[h]
 	\includegraphics[width= 0.7\textwidth, trim = 0 220 0 220]{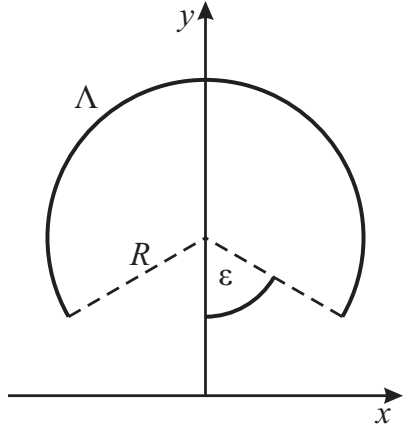}\\
	 \caption{The circular arc of radius $R >0$ subtending the angle
	 $2\pi - 2\eps$ with $\eps \in (0,\pi)$.}\label{fig2}
\end{figure}
\begin{example}\label{example3}
{\rm 
	The circular arc (see Figure~\ref{fig2}) 
	can be efficiently parametrized as follows:
	\begin{equation}\label{arc}
		\Lambda := 
		\big\{(R\sin\varphi, R(1-\cos\varphi)\in\dR^2
			\colon\varphi\in(\eps, 2\pi-\eps)
		\big\},
	\end{equation}
	where $\eps \in (0,\pi)$ and
	$R > 0$ is the radius of the underlying circle.
	Consider the LFT $M(z) := 1/z$. One easily sees that
	\[	
		M_1(x,y) = \Re M(x+\ii y) = \frac{x}{x^2+y^2},
	\]
	and according to~\eqref{J} the Jacobian $\J$ 
	of this LFT is given by the formula
	\begin{equation}\label{Jarc}
	\begin{split}
		\J(x,y) & = \big( (\p_x M_1)^2 + (\p_y M_1)^2 \big)(x,y) \\
				& = \frac{(x^2-y^2)^2}{(x^2+y^2)^4} + \frac{4x^2y^2}{(x^2+y^2)^4}
				 = \frac{1}{(x^2+y^2)^2}.
	\end{split}			
	\end{equation}
	Next observe that $M(\infty) = M^{-1}(\infty) =  0 \notin \Lambda$.
	Moreover, this LFT is inverse to itself
	and under the LFT $M^{-1}(z) = 1/z$
	the arc $\Lambda\subset\dR^2$ is mapped onto 
	the interval 
	\[	
		\Gamma := M^{-1}(\Lambda) 
		=	
		\big\{ \big(x, -\tfrac{1}{2R}\big) \in\dR^2 \colon 
								| x | < \cot(\eps/2)/ (2R) 
		\big\},
	\]
	which is obviously monotone in the sense of Definition~\ref{def:monotone}.
	Compute further $\omega_*(\Gamma)$ defined in \eqref{omega*} 
	\[
		\omega_*(\Gamma) 
		= 
		\inf_{r\in(0,|\Gamma|)} \frac{1}{2\pi r}
		=
		\frac{1}{2\pi |\Gamma|} 
		= 
		\frac{R}{2\pi\cot(\eps/2)};
	\]
	here $|\Gamma|$ in the length of $\Gamma$. Moreover, we obtain
	from \eqref{Jarc} that
	\[
		\sup_{\zz\in\Gamma} \sqrt{\J(\zz)} = 4R^2.
	\]
	Hence, Theorem~\ref{thm:absenceLFT} implies that 
	\[
		\s(\OpIId)  = [0,+\infty)\qquad\text{for all}~ 
		\omega\in \big(-\infty,	(8\pi R)^{-1} \tan(\eps/2) \big].
	\]	
	}
	\end{example}

\section{Remarks and open questions}
\label{sec:open}
In the present paper 
we have analysed
from various perspectives a new effect of absence of the negative spectrum
for Hamiltonians with $\delta^\pp$-interaction supported on 
non-closed curves in $\dR^2$. 
Quite a few questions remain open and we wish to formulate two of them.

Comparing Example~\ref{example1} and Proposition~\ref{prop:interval}
one may pose the following question.
\begin{openA}
	Let the constant $L > 0$ be fixed
	and the interval $\Lambda$ be as in~\eqref{LambdaL}. 
	The problem is to find
	the critical strength $\omega_{\rm cr}(L) > 0 $ 
	such that the operator $\OpIId$ is non-negative
	if and only if $\omega\in\dR$ satisfies $\omega \le \omega_{\rm cr}(L)$.
\end{openA}
The same question as above can be asked for other shapes of $\Lambda$,
but the authors do not expect that an exact formula for the critical
strength can be found.

On one hand, our method of the proof does not allow to cover
curves of generic shape. On the other hand, despite
many attempts, we have not found out
any example of a bounded non-closed curve, 
for which bound states in the weak coupling regime do exist.
A  general open question can be posed.
\begin{openB}
	Is it true that for any bounded sufficiently smooth non-closed curve 
	$\Lambda\subset\dR^2$ there exists a constant $\omega_* > 0$
	such that $\s(\OpIId) = [0,+\infty)$
	for all $\omega\in (-\infty,\omega_*]$?
\end{openB}
It is worth noting that the program carried out in Subsection~\ref{ssec:abs_LFT}
for linear fractional transformations can be generalized by means of Neumann bracketting to arbitrary conformal maps. This could be a possible way to answer
Question B.

Finally, we mention that 
several assumptions play only technical role and can be removed
with additional efforts. Namely, assuming  that $\Lambda$ is a subarc
of the boundary of a Lipschitz domain is technical
as well assuming that the curve $\Lambda$ is
piecewise-$C^1$ in some of the formulations
instead of just being Lipschitz.

\begin{appendix}

\section{Proofs of Lemmas~\ref{lem:change1} and~\ref{lem:change2}}
\label{app:A}

\begin{proof}[Proof of Lemma~\ref{lem:change1}]
	Consider the following two open sets 
	\[
		E := \dR^2 \setminus (\{M(\infty)\}\cup M(\Lambda))
		\quad\text{and}\quad
		F := \dR^2 \setminus  \Lambda.
	\]
	By the formula in Lemma~\ref{lem:M}
	and using that $\dR^2 \setminus E$ is a null set we get
	\[
		\int_{\dR^2} |\nabla v|^2\dd x 
		=  
		\int_E |\nabla v|^2\dd x 
		= 
		\int_E |(\nabla u)\circ M|^2\J\dd x.
	\]
	According to Proposition~\ref{prop:LFT} we have that
	$M^{-1} \colon E\rightarrow F$ 	is a bijection
	which is additionally everywhere differentiable in $E$; \cf
	\eqref{CR}.	Hence, we can apply the substitution rule 
	for Lebesgue integrals (\eg~\cite[Thm. 8.21, Cor. 8.22]{Leoni}) 
	and get
	\[
	\begin{split}
		\int_{\dR^2} |\nabla v|^2\dd x 
		& = 
		\int_F | (\nabla u)\circ M\circ M^{-1} |^2 \J (\J)^{-1} \dd x \\
		& = 
		\int_F |\nabla u|^2 \dd x
		=
		\int_{\dR^2} |\nabla u|^2\dd x;
	\end{split}	
	\]
	where in the last step we employed that $\dR^2\setminus F$ is a null set.
\end{proof}

\begin{proof}[Proof of Lemma~\ref{lem:change2}]
	Observe first that by definition of the curvilinear integral we have
	\begin{equation}\label{wtomegauv}		
		\big(\wt\omega [v]_\Gamma,[v]_\Gamma\big)_\Gamma	
		=	
		\int_0^L \wt\omega( \gamma(s) ) 
		|v_+(\gamma(s)) - v_-(\gamma(s)) |^2 |\gamma^\prime(s)| \dd s.
	\end{equation}
	Using elementary composition rules we also note
	\begin{equation}\label{vu}
	\begin{split}
		v_+( \gamma(s) ) - v_-( \gamma(s) )
		&= (u_+ \circ M \circ M^{-1} \circ \lm)(s) -
		(u_- \circ M \circ M^{-1} \circ \lm)(s) \\
		& = 
		u_+(\lm(s)) - u_-(\lm(s)),
	\end{split}	
	\end{equation}
	where $u_\pm =  u\uhr\Omega_\pm$ and $v_\pm = v\uhr M^{-1}(\Omega_\pm)$.
	Observe also that $\lambda = M\circ \gamma$.
	Using~\eqref{J} and~\eqref{relation} we obtain
	\[
	\begin{split}
		|\lm^\prime(s)|^2
		& =
		\big( (\nabla M_1\circ\gamma)(s) \cdot \gamma^\prime(s) \big)^2
		+
		\big( (\nabla M_2\circ\gamma)(s) \cdot \gamma^\prime(s) \big)^2\\
		& = 
		|(\nabla M_1\circ\gamma)(s)|^2 \cdot |\gamma^\prime(s)|^2\cos^2\alpha
		+
		|(\nabla M_2\circ\gamma)(s)|^2 \cdot |\gamma^\prime(s)|^2\sin^2\alpha\\
		& = 
		\J (\gamma(s)) \cdot |\gamma^\prime(s)|^2,
	\end{split}		
	\]
	where $\alpha$ is the angle between $\nabla M_1$
	and $\gamma^\prime$. Thanks to \eqref{wtomegauv} and using 
	\eqref{vu}	we arrive at
	\[
	\begin{split}
		\big( \wt\omega [v]_\Gamma,	[v]_\Gamma \big)_\Gamma
		 & = 
 		\int_0^L \wt\omega(\gamma(s))
			|v_+(\gamma(s)) - v_-(\lm(s))|^2\, |\gamma^\prime(s)|\dd s\\
		& = 
		\int_0^L \frac{\wt\omega( \gamma(s) )}{\sqrt{\J(\gamma(s))}}\,
			|u_+(\lm(s)) - u_-(\lm(s))|^2\, |\lm^\prime(s)|\dd s.
	\end{split}		
	\]
	Finally, employing~\eqref{def:wt_omega} we end up with
	the desired relation
	\[	
		\big(\wt\omega [v]_\Gamma,	[v]_\Gamma \big)_\Gamma
			= \big(\omega [u]_\Lambda, [u]_\Lambda\big)_\Lambda.
	\]
\end{proof}

\end{appendix}


\def\cprime{$'$} \def\cprime{$'$}

\end{document}